\documentclass[10pt,aps,pra,showpacs,showkeys,twocolumn]{revtex4-1}
\usepackage{amsmath, amssymb}
\usepackage{amsthm}
\usepackage{graphicx}
\usepackage{cases}
\usepackage{bbold}
\usepackage{xcolor}

\newtheorem{thrm}{Theorem}
\newtheorem{lmn}{Lemma}
\newtheorem{prop}{Proposition}

\DeclareMathOperator{\diag}{diag}

\renewcommand{\vec}[1]{\mathbf{#1}}

\begin{document}

\title{Tripartite separability conditions exponentially violated by Gaussian states}
\author{E. Shchukin}
\email{evgeny.shchukin@gmail.com}
\author{P. van Loock}
\email{loock@uni-mainz.de}
\affiliation{Johannes-Gutenberg University of Mainz, Institute of Physics, Staudingerweg 7, 55128
Mainz}

\begin{abstract}
Starting with a set of conditions for bipartite separability of arbitrary quantum states in any 
dimension and expressed in terms of arbitrary operators whose commutator is a $c$-number, we derive 
a hierarchy of conditions for tripartite separability of continuous-variable three-mode quantum 
states. These conditions have the form of inequalities for higher-order moments of linear 
combinations of the mode operators. They enable one to distinguish between all possible kinds of 
tripartite separability, while the strongest violation of these inequalities is a sufficient 
condition for genuine tripartite entanglement. We construct Gaussian states for which the violation 
of our conditions grows exponentially with the order of the moments of the mode operators. By going 
beyond second moments, our conditions are expected to be useful as well for the detection of 
tripartite entanglement of non-Gaussian states. We also demonstrate that our conditions can be 
easily implemented experimentally.
\end{abstract}



\maketitle

\section{Introduction} 

One of the first algorithmic solutions of the problem of characterizing entangled states was given 
in \cite{PhysRevA.69.022308}. The algorithm presented in that work is based on numerical 
optimization and thus inherits all limitations of the numerics it is based on. Hence there is 
clearly a desire for more robust and simple analytical entanglement criteria, which may be adapted 
to specific applications. The most well-known, analytical entanglement criteria are based on a 
certain positive, but not completely positive map -- the partial transposition (PT) of density 
matrices \cite{PhysRevLett.77.1413}. For discrete-variable, bipartite two-qubit and qubit-qutrit 
states \cite{Horodecki19961}, as well as for continuous-variable, bipartite Gaussian states split 
into one versus arbitrarily many modes \cite{PhysRevLett.86.3658, PhysRevLett.84.2726}, negative PT 
(NPT) is necessary and sufficient for entanglement. While NPT is generally only a sufficient 
entanglement condition, as such, it is still an unambiguous certifier for entanglement.

In the multipartite case, of course, NPT can also serve as a general indicator of inseparability, 
provided that the richer structure of systems composed of more than two subsystems is taken into 
account. More specifically, a multipartite state may have NPT with respect to only one particular 
splitting and have positive PT (PPT) otherwise. While such a scenario is ambiguous for more than two 
qubits, it would unambiguously identify a non-fully, only partially inseparable Gaussian state of 
three harmonic oscillators (three optical modes). In fact, arbitrary tripartite, three-mode mixed 
Gaussian states can be completely classified by five classes ranging from fully inseparable to fully 
separable states \cite{PhysRevA.64.052303}. The fully inseparable class may then be considered the 
class of genuinely tripartite entangled states. For discrete-variable states such as three-qubit 
states as well as for three modes in a general continuous-variable, mixed state, the situation is 
more complex, as there are fully inseparable states which can be constructed as convex combinations 
of different partially separable (biseparable) states. The strongest notion of tripartite 
entanglement should exclude such biseparable mixtures. The term genuine multipartite entanglement 
has been used in the literature on continuous-variable entanglement both in the weak and in the 
strong sense: in Ref.~\cite{PhysRevA.67.052315}, conditions for full inseparability were presented 
in terms of linear combinations of the position and momentum operators of arbitrarily many modes 
(applied experimentally for the first time in Ref.~\cite{PhysRevLett.91.080404}), extending the 
previously given bipartite criteria by Duan {\it et al.} \cite{PhysRevLett.84.2722} to the 
multipartite case; in the combined theoretical and experimental work of Ref.~\cite{NatPhys.9.19}, a 
stronger version of the tripartite inequalities of Ref.~\cite{PhysRevA.67.052315} was derived and 
applied; and Hyllus and Eisert \cite{NewJPhys.8.51} gave a unified treatment including both weak and 
strong conditions for multipartite entanglement based on a complete knowledge of the covariance 
matrix of the continuous-variable state in question. All these criteria are based on the second 
moments of the mode operators, and most of the criteria are also, at least implicity, based on PT 
(see also Ref.~\cite{PhysRevLett.111.110503}).

Here, starting from a pair of arbitrary operators $\hat{X}$ and $\hat{Y}$ where $[\hat{X}, \hat{Y}]$ 
is a nonzero $c$-number, we develop a hierarchy of conditions based on PT which all tripartite 
biseparable states \cite{PhysRevLett.87.040401} must satisfy. Each condition is an inequality for 
higher-order moments of the state under consideration. If a tripartite state violates an inequality 
of our hierarchy then it is genuinely entangled in the strong sense. Our hierarchy has the 
remarkable property that the possible violations of our conditions, scaled to the bound for 
biseparable states, grows exponentially with the order of the moments. We show that there exist pure 
three-mode Gaussian states which then give an exponential violation with increasing order. Hence 
such Gaussian states, easily producible with squeezed light and beam splitters 
\cite{PhysRevA.71.055801, PhysRevLett.73.58}, appear to exhibit the strongest possible form of 
tripartite entanglement. Moreover, since our criteria are based on higher-order moments, they can 
also be expected to allow for more sensitive tests of non-Gaussian tripartite entanglement than 
those based only on second moments. Among other approaches to tripartite entanglement we refer to 
Refs.~\cite{PhysRevLett.112.010401, EPJSP.160.319}. Higher-order moments were already used in 
Ref.~\cite{PhysRevA.74.030302} to test whether one group of modes can be separated from the rest. 
Our conditions here are more powerful in the sense that they are valid for arbitrary mixtures of 
differently separable states and thus the violation of these conditions is sufficient for genuine 
tripartite entanglement. The multipartite higher-order moments can be measured experimentally, see 
Refs.~\cite{PhysRevLett.96.200403, PhysRevA.85.063835}.

\section{Bipartite case} 

We start with the following simple lemma, which is usually used for the case of annihilation and 
creation operators to express antinormally ordered moments in terms of normally ordered ones.
\begin{lmn}
For any two operators $\hat{X}$ and $\hat{Y}$ whose commutator $[\hat{X}, \hat{Y}] = c$ is a 
$c$-number, we have
\begin{equation}\label{eq:XY}
    \hat{X}^n \hat{Y}^m = \sum^{\min(n,m)}_{k=0} k! \binom{n}{k} \binom{m}{k} c^k
    \hat{Y}^{m-k} \hat{X}^{n-k}.
\end{equation}
\end{lmn}

With this lemma we can now prove one of our main results, on which our condition for genuine 
tripartite entanglement will be based.
\begin{thrm}\label{thrm:AB}
Let operators $\hat{A}$ and $\hat{B}$ act on different degrees of freedom, with the commutator 
$[\hat{B}, \hat{B}^\dagger] = c > 0$ being a positive $c$-number. Then the inequality
\begin{equation}\label{eq:ABn}
    f_n(\hat{A}, \hat{B}) \equiv \langle (\hat{A}^\dagger + \hat{B})^n
    (\hat{A} + \hat{B}^\dagger)^n \rangle \geqslant c^n n!
\end{equation}
is valid for all bipartite separable quantum states.
\end{thrm}

If, in addition to the assumption of the theorem, the equality $[\hat{A}, \hat{A}^\dagger] = 
[\hat{B}, \hat{B}^\dagger] = c$ holds, then the operators $\hat{A}^\dagger + \hat{B}$ and $\hat{A} + 
\hat{B}^\dagger$ commute, since in this case we have $[\hat{A}^\dagger + \hat{B}, \hat{A} + 
\hat{B}^\dagger] = [\hat{A}^\dagger, \hat{A}] + [\hat{B}, \hat{B}^\dagger] = 0$, and 
Eq.~\eqref{eq:ABn} becomes $\langle \bigl((\hat{A}^\dagger + \hat{B})(\hat{A} + 
\hat{B}^\dagger)\bigr)^n \rangle \geqslant n! c^n$. In the bipartite two-mode case, with $\hat{A} = 
\hat{a}$ and $\hat{B} = -\hat{b}$ being the annihilation operators of the first and the second mode, 
respectively, we have 
\begin{equation}
    2(\hat{a}^\dagger - \hat{b})(\hat{a} - \hat{b}^\dagger) = (\hat{x}_a - 
\hat{x}_b)^2 + (\hat{p}_a + \hat{p}_b)^2 = \hat{O}_{\mathrm{EPR}}, 
\end{equation}
where $\hat{x} = (\hat{a}+\hat{a}^\dagger)/\sqrt{2}$, $\hat{p} = 
-i(\hat{a}-\hat{a}^\dagger)/\sqrt{2}$, and this inequality can be written as $\langle 
\hat{O}^n_{\mathrm{EPR}} \rangle \geqslant 2^n n!$, which has been derived in 
\cite{PhysRevLett.108.030503} using another approach (with different convention for position and 
momentum).

\begin{proof}
Expanding the left-hand side of Eq.~\eqref{eq:ABn} we get 
\begin{equation}
    f_n(\hat{A}, \hat{B}) = \sum^n_{k,l=0} 
\binom{n}{k} \binom{n}{l} \langle\hat{A}^{\dagger n-k} \hat{A}^{n-l} \hat{B}^k \hat{B}^{\dagger l} 
\rangle. 
\end{equation}
Using the relation $\langle \hat{A} \hat{B} \rangle_{\mathrm{PT}} = \langle \hat{A} 
\hat{B}^\dagger \rangle$, which is proven in Appendix A, we have
\begin{equation}
    f_n(\hat{A}, \hat{B}) = \sum^n_{k,l=0} \binom{n}{k} \binom{n}{l} 
\langle\hat{A}^{\dagger n-k} \hat{A}^{n-l} \hat{B}^l \hat{B}^{\dagger k} \rangle_{\mathrm{PT}}. 
\end{equation}
One can easily see that this sum can be simplified as
\begin{equation}\label{eq:AB}
    f_n(\hat{A}, \hat{B}) = \sum^n_{k=0} \binom{n}{k} \langle \hat{A}^{\dagger n-k}
    (\hat{A}+\hat{B})^n \hat{B}^{\dagger k} \rangle_{\mathrm{PT}}.
\end{equation}
Since $[\hat{A}+\hat{B}, \hat{B}^\dagger] = c$ is a $c$-number, we can use the relation
\eqref{eq:XY} to write the product $(\hat{A}+\hat{B})^n \hat{B}^{\dagger k}$ in the
``normally-ordered" form:
\begin{equation}
    (\hat{A}+\hat{B})^n \hat{B}^{\dagger k} = \sum^k_{j=0} j! \binom{n}{j}\binom{k}{j} c^j
    \hat{B}^{\dagger k-j} (\hat{A}+\hat{B})^{n-j}.
\end{equation}
Substituting this expansion into Eq.~\eqref{eq:AB}, after some simplifications we obtain
\begin{equation}\label{eq:ABAB}
\begin{split}
    f_n(\hat{A}, \hat{B})
    = \sum^n_{i=0} \frac{n!}{i!} \binom{n}{i} c^{n-i} \langle (\hat{A}^\dagger + \hat{B}^\dagger)^i
    (\hat{A} + \hat{B})^i \rangle_{\mathrm{PT}}. \nonumber
\end{split}
\end{equation}
For a separable state, the partially transposed state is positive definite, and then the terms on 
the right-hand side of this equality with $i > 0$ are non-negative and we finally get the desired 
inequality \eqref{eq:ABn} \footnote{There is an alternative way of proving this theorem that does 
not use PT explicitly. But nevertheless it turns out that the theorem is PT-related --- if a state 
is PPT then the inequality \eqref{eq:ABn} cannot be violated.}.
\end{proof}

\section{Tripartite case} 

Consider the tripartite case with the three parts labeled by the letters $a$, $b$, and $c$. A state 
$\hat{\varrho}$ is called $1|23$ partially separable if its first part, $a$, can be separated from 
the other two, $b$ and $c$, i.e., if it can be represented as $\hat{\varrho} = \sum_i p_i 
\hat{\varrho}_{a, i} \otimes \hat{\varrho}_{bc, i}$, where $\hat{\varrho}_{a, i}$ are density 
operators of the $a$-part and $\hat{\varrho}_{bc, i}$ are bipartite density operators of the $b$ and 
$c$ parts. To test for this kind of separability we introduce the quantity 
$\mathcal{A}^{(m)}_{1|23}$ defined as
\begin{equation}\label{eq:A}
    \mathcal{A}^{(m)}_{1|23} = \frac{1}{m!} \langle (\hat{a}^\dagger + \hat{b} + \hat{c})^m
    (\hat{a} + \hat{b}^\dagger + \hat{c}^\dagger)^m \rangle.
\end{equation}
This quantity can distinguish between $1|23$ separable states and all others according to the 
following proposition.
\begin{prop} The quantity $\mathcal{A}^{(m)}_{1|23}$ satisfies the inequalities
\begin{subnumcases}{\label{eq:A1-23} \mathcal{A}^{(m)}_{1|23} \geqslant}
    2^m & \text{for all $1|23$ separable states} \label{eq:A1-23-1} \\
    1 & \text{for all states} \label{eq:A1-23-2}
\end{subnumcases}
For the pure Gaussian state $|\psi\rangle$ with wave function $\psi(\vec{x}) = \pi^{-3/4} 
\exp\left(-\frac{1}{2}\vec{x}^{\mathrm{T}} A \vec{x}\right)$, where the matrix is given by $A = 
\left(\begin{smallmatrix} 3 & 2 & 2 \\ 2 & 2 & 1 \\ 2 & 1 & 2 \end{smallmatrix}\right)$, we have 
$\mathcal{A}^{(m)}_{1|23} = 1$ for all $m \geqslant 1$, so this state minimizes all the quantities 
$\mathcal{A}^{(m)}_{1|23}$, $m \geqslant 1$, simultaneously.
\end{prop}
\begin{proof}
Taking $\hat{A} = \hat{a}$ and $\hat{B} = \hat{b} + \hat{c}$ in Theorem \ref{thrm:AB}, we get the 
inequality \eqref{eq:A1-23-1}. On the other hand, the operators $\hat{a}^\dagger + \hat{c}$ and 
$\hat{a} + \hat{c}^\dagger$ commute with $\hat{b}$ and $\hat{b}^\dagger$, and with each other, so we 
have $[\hat{b} + \hat{a}^\dagger + \hat{c}, \hat{b}^\dagger + \hat{a} + \hat{c}^\dagger] = 1$, and 
the inequality \eqref{eq:A1-23-2} follows from Eq.~\eqref{eq:XY}. From the same Eq.~\eqref{eq:XY} we 
have
\begin{equation}\label{eq:mabc}
\begin{split}
    &(\hat{a}^\dagger + \hat{b} + \hat{c})^m (\hat{a} + \hat{b}^\dagger + \hat{c}^\dagger)^m =
     m! \\
    &+ \sum^m_{k=1} (m-k)! \binom{m}{k}^2
    (\hat{a} + \hat{b}^\dagger + \hat{c}^\dagger)^k (\hat{a}^\dagger + \hat{b} + \hat{c})^k.
\end{split}
\end{equation}
It is an easy exercise to verify that $(\hat{a}^\dagger + \hat{b} + \hat{c})|\psi\rangle = 0$, so 
that we have $\mathcal{A}^{(m)}_{1|23} = 1$ for all $m \geqslant 1$.
\end{proof}

It is clear that complete separability, which can be denoted as $1|2|3$ separability, is a stronger 
property than just $1|23$ partial separability, but the inequalities \eqref{eq:A1-23} do not allow 
us to distinguish between these two cases. To overcome this disadvantage, we introduce the 
quantities $\mathcal{A}^{(m)}_{2|13}$ and $\mathcal{A}^{(m)}_{3|12}$ that are obtained from 
$\mathcal{A}^{(m)}_{1|23}$ by an appropriate permutation of $\hat{a}$, $\hat{b}$, and $\hat{c}$. 
These quantities satisfy inequalities similar to inequality \eqref{eq:A1-23} for the corresponding 
kinds of partial separability. We then introduce the symmetric sum $S^{(m)}$ defined as
\begin{equation}
    S^{(m)}_3 = \frac{1}{3}(\mathcal{A}^{(m)}_{1|23} + \mathcal{A}^{(m)}_{2|13} +
    \mathcal{A}^{(m)}_{3|12}).
\end{equation}
Our main result is the following theorem about the properties of this sum.
\begin{thrm}
The sum $S^{(m)}_3$ satisfies the following inequality:
\begin{subnumcases}{\label{eq:Sm} S^{(m)}_3 \geqslant}
        2^m & \text{for fully separable states} \label{eq:Sm:1} \\
        \frac{2^m+2}{3} & \text{for biseparable states} \label{eq:Sm:2} \\
        1 & \text{for all states} \label{eq:Sm:3}
\end{subnumcases}
The condition $S^{(m)}_3 < (2^m+2)/3$ is thus sufficient for genuine tripartite entanglement. For 
any $m \geqslant 1$ there is no physical state with $S^{(m)}_3 = 1$, so the inequality 
\eqref{eq:Sm:3} is always the strict inequality $S^{(m)}_3 > 1$. However, there is a family of pure 
Gaussian states, parametrized with a real parameter $0 \leqslant \xi < 1$, for which $S^{(m)}_3(\xi) 
\to 1$ as $\xi \to 1$ for all $m \geqslant 1$. In this limit we then have a gap in $S^{(m)}_3$ 
values between biseparable and genuinely entangled states exponentially growing with $m$.
\end{thrm}
\begin{proof}
The inequalities \eqref{eq:Sm:1} and \eqref{eq:Sm:3} immediately follow from the inequalities 
\eqref{eq:A1-23-1} and \eqref{eq:A1-23-2}, respectively. To prove the inequality \eqref{eq:Sm:2}, 
recall that a state is biseparable if it is a mixture of $1|23$, $2|13$ and $3|12$ partially 
separable states. For a $1|23$ partially separable state the first term in the definition of 
$S^{(m)}_3$ is larger than or equal to $2^m$ and the other two terms are bounded by $1$ from below, 
so the inequality \eqref{eq:Sm:2} is valid in this case. A similar argument shows that it  also 
holds for all $2|13$ and $3|12$ separable states. But then it is also valid for all their mixtures, 
i.e., for all biseparable states.

If for some $m_0$ there is a pure state $|\psi_0\rangle$ with $S^{(m_0)}_3 = 1$, then we must have 
$\mathcal{A}^{(m_0)}_{1|23} = \mathcal{A}^{(m_0)}_{2|13} = \mathcal{A}^{(m_0)}_{3|12} = 1$. From 
Eq.\eqref{eq:mabc} it then follows that $|\psi_0\rangle$ must be an eigenstate of $\hat{a}^\dagger + 
\hat{b} + \hat{c}$, $\hat{a} + \hat{b}^\dagger + \hat{c}$ and $\hat{a} + \hat{b} + \hat{c}^\dagger$ 
with eigenvalue zero. We show that there is no such common eigenstate even for any two of these 
three operators. If, on the contrary, $|\psi_0\rangle$ is a zero-eigenstate of $\hat{a}^\dagger + 
\hat{b} + \hat{c}$ and $\hat{a} + \hat{b}^\dagger + \hat{c}$, then 
\begin{equation}
(\hat{a}^\dagger + \hat{b} + \hat{c}) |\psi_0\rangle = (\hat{a} + \hat{b}^\dagger + \hat{c}) 
|\psi_0\rangle = 0 
\end{equation}
and, thus, we have $(\hat{a}^\dagger - \hat{b}^\dagger) |\psi_0\rangle = (\hat{a} - \hat{b}) 
|\psi_0\rangle$. In other words, we have $\hat{X}^\dagger |\psi_0\rangle = \hat{X} |\psi_0\rangle$, 
where $\hat{X} = \hat{a} - \hat{b}$ and $[\hat{X}, \hat{X}^\dagger] = 2$. From this we get the 
equality $\langle \psi_0 | \hat{X}^2 |\psi_0\rangle - \langle \psi_0 | \hat{X}^{\dagger 2} 
|\psi_0\rangle = 2$, which is contradictory to the fact that the left-hand side is purely imaginary 
or zero. Hence we have just proved that the equality $S^{(m_0)}_3 = 1$ is impossible and thus the 
strict inequality $S^{(m_0)}_3 > 1$ is always valid for all pure states. But due to linearity of 
$S^{(m_0)}_3$ this strict inequality is valid for all, including mixed, tripartite quantum states.

To show that the inequalities $S^{(m)}_3 > 1$, $m \geqslant 1$ are tight we must find a family of
states for which $S^{(m)}_3$ goes arbitrarily close to $1$. Such a family is given by the pure
Gaussian states $|\psi_\xi\rangle$ with wave function 
\begin{equation}
    \psi_\xi(\vec{x}) = \pi^{-3/4} \sqrt[4]{\det A_\xi} \exp\left(-\frac{1}{2}\vec{x}^{\mathrm{T}} 
A_\xi \vec{x}\right), 
\end{equation}
where the matrix $A_\xi$ is defined as 
\begin{equation}
    A_\xi = 
    \begin{pmatrix} 
        1 & \xi & \xi \\ 
        \xi & 1 & \xi \\ 
        \xi & \xi & 1
    \end{pmatrix}. 
\end{equation}
This matrix is positive definite for $-1/2 < \xi < 1$ and in Appendix B we show that 
$S^{(m)}_3(\xi) \to 1$ when $\xi \to 1$. 
\end{proof}

One can also give an example of a physical state for which the violation of the 
inequalities \eqref{eq:Sm} grows exponentially, though at a more moderate rate.
Consider the so-called GHZ/W states $|\varphi_a\rangle$, which are pure 
Gaussian states with the matrix 
\begin{equation}
    A_a = 
    \begin{pmatrix} 
        a & e_- & e_- \\ 
        e_- & a & e_- \\ 
        e_- & e_- & a 
    \end{pmatrix}, 
\end{equation}
where, according to Refs.~\cite{PhysRevA.73.032345, 
JModOpt.50.801, PhysRevLett.84.3482}, 
\begin{equation}
    e_\pm = \frac{a^2-1\pm\sqrt{(a^2-1)(9a^2-1)}}{4a}.
\end{equation}
This matrix is positive-definite for $a>1$. The choice of the off-diagonal elements $e_-$ is 
motivated by the desire to have the diagonal elements of the inverse matrix 
\begin{equation}
    A^{-1}_a = 
    \begin{pmatrix} 
        a & e_+ & e_+ \\ 
        e_+ & a & e_+ \\ 
        e_+ & e_+ & a 
    \end{pmatrix}
\end{equation}
equal to the diagonal elements of $A_a$. The covariance matrix of $|\varphi_a\rangle$ is 
\begin{equation}
    \gamma_a = 
    \begin{pmatrix} 
        A^{-1}_a & 0 \\ 
        0 & A_a 
    \end{pmatrix}, 
\end{equation}
so such a state is symmetric not only with respect to its modes, but up to some extent also with 
respect to position and momentum operators. In Appendix C it is shown that for the state of 
this family with $a = 3/2$ the violation of our condition grows at least as $1.16^m$.

According to the classification of tripartite Gaussian states from Ref.~\cite{PhysRevA.64.052303} 
the states $|\psi_\xi\rangle$ and $|\varphi_a\rangle$ belong to class 1, the class of states which 
are not separable for any grouping of the parties.  A tripartite Gaussian state belongs to class 1 
if the partially transposed matrices no longer satisfy the physicality condition, i.e., if 
$\Lambda_1 \gamma \Lambda_1 \not\geqslant iJ$, $\Lambda_2 \gamma \Lambda_2 \not\geqslant iJ$ and 
$\Lambda_3 \gamma \Lambda_3 \not\geqslant iJ$, where $\Lambda_1 = \diag(1, 1, 1, -1, 1, 1)$, 
$\Lambda_2 = \diag(1, 1, 1, 1, -1, 1)$, $\Lambda_3 = \diag(1, 1, 1, 1, 1, -1)$ and 
\begin{equation}
    J = 
    \begin{pmatrix} 
        0 & -\mathbb{1} \\ 
        \mathbb{1} & 0 
    \end{pmatrix}. 
\end{equation}
Since the state $|\psi_\xi\rangle$ is symmetric, it is enough to test only one of the three 
conditions. We show that $\Lambda_1 \gamma \Lambda_1 \not\geqslant iJ$. The covariance matrix 
$\gamma_\xi$ of the state $|\psi_\xi\rangle$ reads as 
\begin{equation}
    \gamma_\xi = 
    \begin{pmatrix} 
        A^{-1}_\xi & 0 \\ 
        0 & A_\xi 
    \end{pmatrix}. 
\end{equation}
The leading principal minor of the matrix $\Lambda_1 \gamma_\xi \Lambda_1 - iJ$ of size 5 is equal 
to $-4\xi^2/(1-\xi)^2(1+2\xi) < 0$ and thus the condition $\Lambda_1 \gamma \Lambda_1 \not\geqslant 
iJ$ is fulfilled. The same minor of the original matrix $\gamma_\xi - iJ$ equals to zero, so there 
is no violation of the physicality condition, as it must be. The principal minor of $\Lambda_1 
\gamma_a \Lambda_1 - iJ$ for $a=1.5$ obtained by removing the second row and column is equal to 
$-1.185 < 0$ and thus $|\varphi_a\rangle$ also belongs to the class 1. In fact, in the case of a 
pure state, one could verify its non-separability directly from its wave function.

\section{Conditions in terms of position and momentum operators}

The lowest-order quantity $S^{(1)}_3$ can be easily expressed in terms of the position and momentum
operators: $S^{(1)}_3 = (1/6) T^{(1)}_3 + (1/2)$, where
\begin{equation}\label{eq:T}
\begin{split}
    T^{(1)}_3 &= 3\langle(\hat{x}_a + \hat{x}_b + \hat{x}_c)^2\rangle +
    \langle(-\hat{p}_a + \hat{p}_b + \hat{p}_c)^2\rangle \\
    &+ \langle(\hat{p}_a - \hat{p}_b + \hat{p}_c)^2\rangle +
    \langle(\hat{p}_a + \hat{p}_b - \hat{p}_c)^2\rangle.
\end{split}
\end{equation}
The inequalities \eqref{eq:Sm} for $m=1$ can be translated to the following inequalities for
$T^{(1)}_3$:
\begin{subnumcases}{\label{eq:Sxp} T^{(1)}_3 \geqslant}
        9 & \text{for fully separable states} \label{eq:Sxp:1} \\
        5 & \text{for biseparable states} \label{eq:Sxp:2} \\
        3 & \text{for all states} \label{eq:Sxp:3}
\end{subnumcases}
Note that the strict inequality $T^{(1)}_3 > 3$ is valid for any quantum state since $S^{(1)}_3 > 1$ 
holds true. We now derive the inequalities \eqref{eq:Sxp} directly from the uncertainty relation.

The left-hand side of these inequalities is the sum of the three quantities 
$\langle(\hat{x}_a + \hat{x}_b + \hat{x}_c)^2 + (-\hat{p}_a + \hat{p}_b + \hat{p}_c)^2\rangle$, 
$\langle(\hat{x}_a + \hat{x}_b + \hat{x}_c)^2 + (\hat{p}_a - \hat{p}_b + \hat{p}_c)^2\rangle$ and 
$\langle(\hat{x}_a + \hat{x}_b + \hat{x}_c)^2 + (\hat{p}_a + \hat{p}_b - \hat{p}_c)^2\rangle$. Each 
quantity can be estimated with the inequality $\langle\hat{X}^2 + \hat{Y}^2\rangle \geqslant 
|\langle[\hat{X}, \hat{Y}]\rangle|$, which is valid for arbitrary observables $\hat{X}$ and 
$\hat{Y}$. This inequality follows from the positivity of $\langle(\hat{X} \pm i\hat{Y})(\hat{X} 
\mp i\hat{Y})\rangle \geqslant 0$. From this we immediately obtain the inequality 
\begin{equation}
    \langle(\hat{x}_a + \hat{x}_b + \hat{x}_c)^2\rangle + \langle(-\hat{p}_a + \hat{p}_b + 
    \hat{p}_c)^2\rangle \geqslant 1, 
\end{equation}
and similar inequalities for the other two quantities, so \eqref{eq:Sxp:3} follows. 

If a state has some separability properties then we can obtain a better estimation. Note that it is 
enough to prove inequalities \eqref{eq:Sxp:1} and \eqref{eq:Sxp:2} only for the corresponding 
factorizable states and then validity of these inequalities for general separable states will follow 
by linearity of the average. For a $1|23$-factorizable state we have 
\begin{equation}
\begin{split}
    \langle(\hat{x}_a + \hat{x}_b + \hat{x}_c)^2\rangle &= \langle (\Delta \hat{x}_a)^2 + (\Delta 
\hat{x}_b + \Delta \hat{x}_c)^2 \rangle \\
&+ \langle \hat{x}_a + \hat{x}_b + \hat{x}_c \rangle^2 \\
\langle(-\hat{p}_a + \hat{p}_b + \hat{p}_c)^2\rangle &= \langle (\Delta \hat{p}_a)^2 + (\Delta 
\hat{p}_b + \Delta \hat{p}_c)^2 \rangle \\
&+ \langle -\hat{p}_a + \hat{p}_b + \hat{p}_c \rangle^2.
\end{split}
\end{equation}
Since $\langle(\Delta \hat{x}_a)^2 + (\Delta \hat{p}_a)^2\rangle \geqslant 1$ and 
\begin{equation}
    \langle(\Delta \hat{x}_b + \Delta \hat{x}_c)^2 + (\Delta \hat{p}_b + \Delta \hat{p}_c)^2\rangle 
\geqslant 2, 
\end{equation}
we get 
\begin{equation}
    \langle(\hat{x}_a + \hat{x}_b + \hat{x}_c)^2\rangle + \langle(-\hat{p}_a + 
    \hat{p}_b + \hat{p}_c)^2\rangle \geqslant 3. 
\end{equation}
The other two quantities can be estimated either as we did in the previous paragraph or using this 
alternative approach. Both lead to the same bound, $1$, so we arrive at inequality 
\eqref{eq:Sxp:2}. 

For a completely factorizable state we have the identity 
\begin{equation}
\begin{split}
    \langle(\hat{x}_a + \hat{x}_b + \hat{x}_c)^2\rangle &= \langle (\Delta \hat{x}_a)^2 + (\Delta 
    \hat{x}_b)^2 + (\Delta \hat{x}_c)^2\rangle \\
    &+ \langle \hat{x}_a + \hat{x}_b + \hat{x}_c \rangle^2, 
\end{split}
\end{equation}
and similar identities 
for the combinations of $\hat{p}$'s. We thus have 
\begin{equation}
    \langle(\hat{x}_a + \hat{x}_b + \hat{x}_c)^2\rangle + \langle(-\hat{p}_a + \hat{p}_b + 
\hat{p}_c)^2\rangle \geqslant 3 
\end{equation}
and similar inequalities for the other two quantities, which establishes inequality 
\eqref{eq:Sxp:1}.

We now demonstrate that the state $|\psi_\xi\rangle$ in the limit $\xi \to 1$ perfectly violates 
the inequality \eqref{eq:Sxp:3}. One can easily see that 
\begin{equation}
    \hat{\mathbf{p}}|\psi_\xi\rangle = iA_\xi\mathbf{x}|\psi_\xi\rangle. 
\end{equation}
This means that the state $|\psi_\xi\rangle$ is a simultaneous 
eigenstate with zero eigenvalue of the non-Hermitian operators $\hat{p}_1 - i\hat{x}_1 - 
i\xi\hat{x}_2 - i\xi\hat{x}_3$, $\hat{p}_2 - i\xi\hat{x}_1 - i\hat{x}_2 - i\xi\hat{x}_3$, and 
$\hat{p}_3 - i\xi\hat{x}_1 - i\xi\hat{x}_2 - i\hat{x}_3$. In the limit $\xi \to 1$ the (unphysical) 
limiting state $|\psi_1\rangle$ is a simultaneous eigenstate with eigenvalue zero of the 
non-Hermitian operators $\hat{p}_k - i\hat{X}$, $k = 1, 2, 3$, $\hat{X} = \hat{x}_1 + \hat{x}_2 + 
\hat{x}_3$. One can show that for such an eigenstate one has $T^{(1)}_3 = 3$, i.e. such an 
eigenstate perfectly violates the inequality \eqref{eq:Sxp:3} and thus the inequality 
\eqref{eq:Sm:3} (which is impossible for any physical state). In fact, in the limit $\xi \to 1$ 
we have the relations 
\begin{equation}
    \hat{p}_1 |\psi_1\rangle = \hat{p}_2 |\psi_1\rangle = \hat{p}_3 |\psi_1\rangle = i\hat{X} 
    |\psi_1\rangle. 
\end{equation}
From these equalities we get $T^{(1)}_3 = 6\langle\hat{X}^2\rangle$. To compute the last average, 
note that 
\begin{equation}
    \hat{p}^2_1 |\psi_1\rangle = -\hat{X}^2 |\psi_1\rangle + |\psi_1\rangle. 
\end{equation}
In a similar way one obtains that $\langle\hat{p}_j \hat{p}_k\rangle = 
-\langle\hat{X}^2\rangle + 1$ for all $j, k = 1, 2, 3$. Then, on the one hand we have 
\begin{equation}
    \sum^3_{j,k=1} \langle\hat{p}_j \hat{p}_k\rangle =  
\langle(\hat{p}_1+\hat{p}_2+\hat{p}_3)^2\rangle = 9\langle\hat{X}^2\rangle. 
\end{equation}
On the other hand, we have $\sum^3_{j,k=1} \langle\hat{p}_j 
\hat{p}_k\rangle = -9\langle\hat{X}^2\rangle + 9$. Comparing both equalities, we obtain 
$\langle\hat{X}^2\rangle = 1/2$ and thus $T^{(1)}_3 = 3$.

\section{Experiment} 

As it is written, the left-hand side of the inequalities \eqref{eq:Sm} (or, equivalently, the 
quantity \eqref{eq:T} and its higher-order extensions) cannot be measured straightforwardly. In 
Eq.~\eqref{eq:T} the combinations of moments commute with each other, but do not commute with the 
combination of positions, so these four quantities cannot be measured simultaneously. We show now 
that with little modification of the inequalities \eqref{eq:Sm} it is possible to avoid this 
drawback.

Let us augment our tripartite system with an auxiliary fourth mode, $z$, which we will assume 
separable from the first three, i.e., we consider fourpartite states of the form $\hat{\varrho}' = 
\sum_i p_i \hat{\varrho}_{abc, i} \otimes \hat{\varrho}_{z, i}$. The quantity 
$\mathcal{A}^{(m)}_{1|23}$, defined by Eq.~\eqref{eq:A}, is generalized in the following way:
\begin{equation}
    \mathcal{A}^{\prime (m)}_{1|23} = \frac{1}{m!} \langle (\hat{a}^\dagger + \hat{b} + \hat{c} +
\hat{z}^\dagger)^m (\hat{a} + \hat{b}^\dagger + \hat{c}^\dagger + \hat{z})^m \rangle.
\end{equation}
Taking $\hat{A} = \hat{a} + \hat{z}$ and $\hat{B} = \hat{b}+\hat{c}$ in \eqref{eq:ABn} we obtain 
that if all $\hat{\varrho}_{abc, i}$ are $1|23$-separable then $\mathcal{A}^{\prime (m)}_{1|23} 
\geqslant 2^m$. On the other hand, if states $\hat{\varrho}_{abc, i}$ are arbitrary, we can take 
$\hat{A} = \hat{z}$ and $\hat{B} = \hat{a}^\dagger + \hat{b} + \hat{c}$ to conclude that 
$\mathcal{A}^{\prime (m)}_{1|23} \geqslant 1$ in this case. We see that $\mathcal{A}^{\prime 
(m)}_{1|23}$ satisfies the inequalities \eqref{eq:A1-23} under the assumption that the fourth mode 
is separable from the first three. Moreover, if $\hat{\varrho}' = \hat{\varrho} \otimes 
|0\rangle\langle 0|$, then $\mathcal{A}^{\prime (m)}_{1|23} = \mathcal{A}^{(m)}_{1|23}$, so if we 
can easily measure $\mathcal{A}^{\prime (m)}_{1|23}$ we can also easily measure 
$\mathcal{A}^{(m)}_{1|23}$.

From the inequalities for the augmented quantities $\mathcal{A}^{\prime (1)}_{1|23}$, 
$\mathcal{A}^{\prime (1)}_{2|13}$ and $\mathcal{A}^{\prime (1)}_{3|12}$ we can obtain an inequality 
in terms of the position and momentum operators similar to the inequality 
\eqref{eq:Sxp}. It reads as follows:
\begin{subnumcases}{\label{eq:T'} T^{\prime (1)}_3 - 3 \geqslant}
        9 & \text{for fully separable states} \label{eq:T':1} \\
        5 & \text{for biseparable states} \label{eq:T':2} \\
        3 & \text{for all states} \label{eq:T':3}
\end{subnumcases}
where $T^{\prime (1)}_3$ is defined by
\begin{widetext}
\begin{equation}\label{eq:T'2}
\begin{split}
    T^{\prime (1)}_3 &= 3\langle(\hat{x}_a + \hat{x}_b + \hat{x}_c + \hat{x}_z)^2\rangle +
    \langle(\hat{p}_a - \hat{p}_b - \hat{p}_c + \hat{p}_z)^2\rangle
    + \langle(- \hat{p}_a + \hat{p}_b - \hat{p}_c + \hat{p}_z)^2\rangle +
    \langle(- \hat{p}_a - \hat{p}_b + \hat{p}_c + \hat{p}_z)^2\rangle \\
    &= \langle (\hat{X}_{az} + \hat{X}_{bc})^2 + (\hat{P}_{az} - \hat{P}_{bc})^2
    + (\hat{X}_{bz} + \hat{X}_{ac})^2 + (\hat{P}_{bz} - \hat{P}_{ac})^2
    + (\hat{X}_{cz} + \hat{X}_{ab})^2 + (\hat{P}_{cz} - \hat{P}_{ab})^2\rangle,
\end{split}
\end{equation}
\end{widetext}
with $\hat{X}_{ij} = \hat{x}_i + \hat{x}_j$ and $\hat{P}_{ij} = \hat{p}_i + \hat{p}_j$, while the 
separability test is applied to the first three modes (since the fourth mode is always assumed to be 
separable from the rest). For a fully separable tripartite state the corresponding four-partite 
state is also fully separable and we can apply the result of Ref.~\cite{PhysRevLett.84.2722} (scaled 
by a factor of 2 since we have a sum of two operators) to get that $\langle (\hat{X}_{ij} + 
\hat{X}_{kl})^2 + (\hat{P}_{ij} - \hat{P}_{kl})^2 \rangle \geqslant 4$, $\{i, j, k, l\} = \{a, b, c, 
z\}$, and thus the inequality \eqref{eq:T':1} is obtained. If $\hat{X}_{ij}$ and $\hat{X}_{kl}$ were 
ordinary single-mode operators then $\langle (\hat{X}_{ij} + \hat{X}_{kl})^2 + (\hat{P}_{ij} - 
\hat{P}_{kl})^2 \rangle$ could be arbitrarily close to zero, but in our case each operator 
$\hat{X}_{ij}$ acts on two modes and the fact that one of the modes is always separable from the 
rest puts a lower bound on the values that this quantity can take. This restriction is expressed by 
$\langle (\hat{X}_{ij} + \hat{X}_{kl})^2 + (\hat{P}_{ij} - \hat{P}_{kl})^2 \rangle \geqslant 2$, so 
we arrive at the inequalities \eqref{eq:T':2} and \eqref{eq:T':3}.

The advantage of introducing a new mode is that all the quantities in Eq.~\eqref{eq:T'2} commute 
with each other and thus can be measured simultaneously. In Fig.~\ref{fig:meas} we present a scheme 
to measure $T^{\prime (1)}_3$, defined by Eq.~\eqref{eq:T'2}. It consists of four 50-50 beam 
splitters and two mirrors. The first three inputs are used for the tripartite state in question and 
the auxiliary fourth input is vacuum. On the output we get precisely the four combinations from 
Eq.~\eqref{eq:T'2}. The higher-order quantities can be measured with a generalization of the 
approach proposed in Ref.~\cite{PhysRevA.66.030301}.

\begin{figure}[ht]
\includegraphics{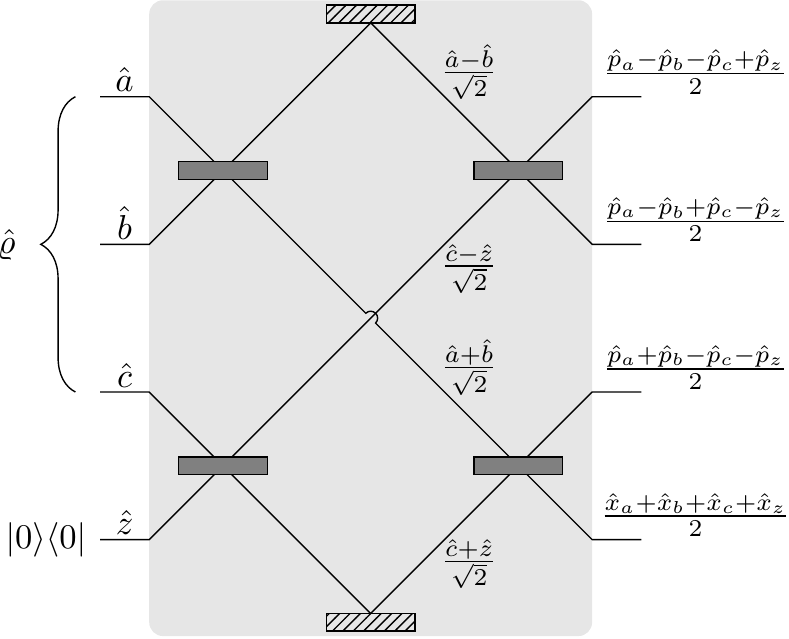}
\caption{Measurement scheme for $T^{\prime (1)}_3$, Eq.~\eqref{eq:T'2}. \label{fig:meas}}
\end{figure}

\section{Conclusion} 

In conclusion, we have derived a hierarchy of conditions for tripartite separability in terms of 
higher-order moments. Violation of these conditions is a clear signature of tripartite genuine 
entanglement. We have constructed tripartite Gaussian states for which the violation of the 
conditions grows exponentially with the order of the moments. Higher-order conditions of our 
hierarchy can potentially be useful to detect entanglement of non-Gaussian states as well. We also 
proposed a method to implement the measurements of our conditions in a straightforward way.

\appendix

\section{Moments of partially transposed state} 

We first prove the relation $\langle\hat{A}\hat{B}\rangle_{\mathrm{PT}} = 
\langle\hat{A}\hat{B}^\dagger\rangle$, which we use to prove Theorem 1. For a density operator
\begin{equation}
    \hat{\varrho} = \sum_{n,m,k,l} \varrho_{nmkl} |n\rangle_1 \langle k| \otimes |m\rangle_2 
\langle l|
\end{equation}
its partial transposition with respect to the second subsystem reads as
\begin{equation}
    \hat{\varrho}^{\mathrm{PT}} = \sum_{n,m,k,l} \varrho_{nmkl} |n\rangle_1 \langle k| \otimes 
|l\rangle_2 \langle m|.
\end{equation}
In these expressions the states $\{|n\rangle_1\}$ form a basis of the Hilbert space associated with 
the first subsystem and similarly $\{|m\rangle_2\}$ is a basis of the second subsystem. We then have
\begin{equation}
    \langle\hat{A}\hat{B}\rangle_{\mathrm{PT}} =  
    \sum_{n,m,k,l} \varrho_{nmkl} \langle k|\hat{A}|n\rangle \langle m|\hat{B}|l\rangle.
\end{equation}
On the other hand, we can write
\begin{equation}
\begin{split}
    \langle\hat{A}\hat{B}^\dagger\rangle &= \sum_{n,m,k,l} \varrho_{nmkl} \langle 
k|\hat{A}|n\rangle \langle l|\hat{B}^\dagger|m\rangle \\
    &= \sum_{n,m,k,l} \varrho_{nmkl} \langle k|\hat{A}|n\rangle \langle m|\hat{B}|l\rangle = 
\langle\hat{A}\hat{B}\rangle_{\mathrm{PT}}.
\end{split}
\end{equation}
Here we have used an additional assumption that the matrix elements $\langle m|\hat{B}|l\rangle$ 
are 
real in the basis $\{|m\rangle_2\}$. In the present work we deal with creation and annihilation 
operators in the Fock basis, so this assumption is valid in our case.

\section{Proof of Theorem 2} 

Here we prove that the inequalities $S^{(m)}_3 > 1$ are tight, i.e., one cannot 
improve the lower bound. Let us take an arbitrary real tripartite Gaussian state $|\psi\rangle$ 
symmetric with respect to its modes. The matrix $A$ of this state corresponds to
\begin{equation}
    A = 
    \begin{pmatrix}
        \mu & \nu & \nu \\
        \nu & \mu & \nu \\
        \nu & \nu & \mu
    \end{pmatrix},
\end{equation}
where $\mu$ and $\nu$ are some real numbers such that $A$ is positive-definite. Due to symmetry all 
three terms in the definition of $S^{(m)}_3$ are equal and we have
\begin{equation}
    S^{(m)}_3 = \frac{1}{m!} \langle (\hat{a}^\dagger + \hat{b} + \hat{c})^m 
                                     (\hat{a} + \hat{b}^\dagger + \hat{c}^\dagger)^m \rangle,
\end{equation}
so it is enough to compute only one of the terms. We define $\hat{Z} = \hat{a} + \hat{b}^\dagger + 
\hat{c}^\dagger$. To find $\hat{Z}^m |\psi\rangle$ we first compute the generating function $e^{t 
\hat{Z}} |\psi\rangle$ and then take the coefficient of $t^m$. In the position representation the 
operator $\hat{Z}$ becomes
\begin{equation}
\begin{split}
    \hat{Z} &= \frac{1}{\sqrt{2}}\left(x + y + z + \frac{\partial}{\partial x} - 
              \frac{\partial}{\partial y} - \frac{\partial}{\partial z}\right) \\
            &= \frac{1}{\sqrt{2}}(\vec{j}^\mathrm{T} \vec{x} + \vec{e}^\mathrm{T} \nabla),
\end{split}
\end{equation}
where we use the notation $\vec{j} = (1, 1, 1)$, $\vec{e} = (1, -1, -1)$ and $\nabla = 
(\partial/\partial x, \partial/\partial y, \partial/\partial z)$. The exponent of $\hat{Z}$ can be 
easily disentangled with help of the Baker-Campbell-Hausdorff formula \cite{*[][{ p. 525, 
Eq.~(11.3-4a)}] mandel-wolf}, and we have the following result:
\begin{equation}
    e^{t\hat{Z}} = e^{-t^2/4} e^{(t/\sqrt{2})\vec{j}^\mathrm{T} \vec{x}} 
                   e^{(t/\sqrt{2})\vec{e}^\mathrm{T} \nabla}.
\end{equation}
The Taylor expansion formula can be compactly written in the following symbolic form:
\begin{equation}
    e^{\vec{v}^\mathrm{T} \nabla} f(\vec{x}) = f(\vec{x} + \vec{v}),
\end{equation}
so that we can easily compute $\langle\vec{x}|e^{t \hat{Z}} |\psi\rangle$ as
\begin{equation}
\begin{split}
    \langle\vec{x}|e^{t \hat{Z}} |\psi\rangle &= e^{-\frac{t^2}{4} + 
    \frac{t}{\sqrt{2}}\vec{j}^\mathrm{T} \vec{x}}
    \psi\left(\vec{x} + \frac{t}{\sqrt{2}}\vec{e}\right) \\
    &= e^{-\frac{\lambda}{4}t^2 + \frac{t}{\sqrt{2}}\vec{u}^\mathrm{T} \vec{x}}
    \psi(\vec{x}),
\end{split}
\end{equation}
where $\lambda = 1+\vec{e}^\mathrm{T}A\vec{e}$ and $\vec{u} = \vec{j} - A\vec{e}$. From this we 
immediately get an explicit expression for $\hat{Z}^m |\psi\rangle$,
\begin{equation}
    \langle\vec{x}| \hat{Z}^m |\psi\rangle = \frac{\sqrt{\lambda^m}}{2^m} 
    H_m\left(\frac{\vec{u}^\mathrm{T} \vec{x}}{\sqrt{2\lambda}}\right) \psi(\vec{x}).
\end{equation}
We then have the following expression for our quantity:
\begin{equation}
\begin{split}
    S^{(m)}_3 &= \frac{1}{m!} \langle \hat{Z}^{\dagger m} \hat{Z}^m \rangle \\
    &= \frac{\lambda^m}{m! 2^{2m}} \int H^2_m\left(\frac{\vec{u}^\mathrm{T} 
    \vec{x}}{\sqrt{2\lambda}}\right) \psi^2(\vec{x}) \, d\vec{x}.
\end{split}
\end{equation}
To compute this integral, we introduce another generating function
\begin{equation}
    G(t) = \sum^{+\infty}_{m=0} S^{(m)}_3 t^m.
\end{equation}
With the help of the relation \cite{*[][{ p. 194, Eq.~(22)}] BE2-2}
\begin{equation}
    \sum^{+\infty}_{m=0} \frac{(s/2)^m}{m!} H^2_m(x) = \frac{1}{\sqrt{1-s^2}} e^{\frac{2s}{1+s}x^2}
\end{equation}
we can write $G(t)$ in the following form:
\begin{equation}
    G(t) = \frac{1}{\sqrt{1-\frac{\lambda^2}{4}t^2}} \sqrt{\frac{\det A}{\det A'}},
\end{equation}
where $A'$ is a rank-one update of the matrix $A$
\begin{equation}
    A' = A - \frac{1}{2} \frac{t}{1+\frac{\lambda}{2}t} \vec{u} \vec{u}^\mathrm{T}.
\end{equation}
For the determinant of this new matrix we have \cite{*[][{ p. 26, Eq.~(0.8.5.11)}] horn-johnson}
\begin{equation}
    \det A' = \det A \left(1-\frac{1}{2}\frac{t}{1+\frac{\lambda}{2}t} 
              \vec{u}^\mathrm{T} A^{-1} \vec{u} \right),
\end{equation}
and then we finally get an explicit expression for the generating function
\begin{equation}\label{eq:G}
    G(t) = \frac{1}{\sqrt{1-\frac{1+\vec{e}^\mathrm{T} A \vec{e}}{2}t}} 
    \frac{1}{\sqrt{1-\frac{1+\vec{j}^\mathrm{T} A^{-1} \vec{j}}{2}t}}.
\end{equation}

We can now prove the rest of Theorem 2. In the case of the state $|\psi_\xi\rangle$ we have
$A \equiv A_\xi =  
    \left(
    \begin{smallmatrix}
        1 & \xi & \xi \\
        \xi & 1 & \xi \\
        \xi & \xi & 1
\end{smallmatrix}\right)$,
and $1+\vec{e}^\mathrm{T} A \vec{e} = 4-2\xi$ and $1+\vec{j}^\mathrm{T} A^{-1} \vec{j} = 
\frac{4+2\xi}{1+2\xi}$, so that
\begin{equation}
    G_\xi(t) = \frac{1}{\sqrt{1-(2-\xi)t}} \frac{1}{\sqrt{1-\frac{2+\xi}{1+2\xi}t}}.
\end{equation}
In the limit $\xi \to 1$ we immediately obtain
\begin{equation}
    G_1(t) = \frac{1}{1-t} = 1 + t + t^2 + t^3 + \ldots,
\end{equation}
which means that $\lim_{\xi \to 1} S^{(m)}_3(\xi) = 1$. 

One can immediately see that our conditions can also be violated for mixed non-Gaussian states of 
the form $(1-p)|\psi_\xi\rangle\langle\psi_\xi| + p|\mathrm{vac}\rangle\langle\mathrm{vac}|$. The 
conditions $S^{(m)}_3 \geqslant (2^m+2)/3$ are violated for $p < 1/3$ in the limit $\xi \to 1$. 
This means that if the vacuum noise is not so strong ($p < 1/3$) then our conditions still work.

By analogy with $\mathcal{A}^{(m)}_{1|23}$ we can introduce a similar quantity 
$\tilde{\mathcal{A}}^{(m)}_{1|23} = \langle (-\hat{a}^\dagger + \hat{b} + \hat{c})^m (-\hat{a} + 
\hat{b}^\dagger + \hat{c}^\dagger)^m \rangle$ and the corresponding symmetric sum 
$\tilde{S}^{(m)}_3$. This new quantity also satisfies the inequalities \eqref{eq:Sm}. In terms of 
position and momentum operators $\tilde{T}^{(1)}_3$, the analogue of $T^{(1)}_3$, reads exactly as 
$T^{(1)}_3$, Eq.~\eqref{eq:T}, but with all position operators interchanged with their corresponding 
momentum operators. It can be shown that if a pure symmetric tripartite Gaussian state with a matrix 
$A$ violates the inequalities for $S^{(m)}_3$ for some $m$ then the state with the inverse matrix 
$A^{-1}$ violates the inequalities for $\tilde{S}^{(m)}_3$ for the same $m$. In fact, we have 
\begin{equation}
    \hat{\tilde{Z}} = -\hat{a} + \hat{b}^\dagger + \hat{c}^\dagger = 
    -\frac{1}{\sqrt{2}}(\vec{e}^\mathrm{T} \vec{x} + \vec{j}^\mathrm{T} \nabla),
\end{equation}
and for the generating function $\tilde{G}(t) = \sum^{+\infty}_{m=0} \tilde{S}^{(m)}_3 t^m$ we 
obtain the following expression:
\begin{equation}
    \tilde{G}(t) = \frac{1}{\sqrt{1-\frac{1+\vec{e}^\mathrm{T} A^{-1} 
\vec{e}}{2}t}} 
    \frac{1}{\sqrt{1-\frac{1+\vec{j}^\mathrm{T} A \vec{j}}{2}t}}.
\end{equation}
We see that $\tilde{G}(t)$ is the same as $G(t)$ with $A \to A^{-1}$. This is a general result 
valid 
for all symmetric tripartite pure Gaussian states with real wave function.

\section{Exponential violation by a physical state} 

For the generating function \eqref{eq:G} we have
\begin{equation}
    G(t) = \frac{1}{\sqrt{1 - \alpha t}} \frac{1}{\sqrt{1 - \beta t}},
\end{equation}
where 
\begin{equation}
    \alpha = \frac{27-\sqrt{385}}{8} \approx 0.92, \quad \beta = 
\frac{61-\sqrt{385}}{24} \approx 
    1.72.
\end{equation}
From this expression we obtain
\begin{equation}
\begin{split}
    &S^{(m)}_3 = \frac{1}{2^{2m}} \sum^{m}_{k=0} \binom{2k}{k} 
\binom{2(m-k)}{m-k} \alpha^k 
    \beta^{m-k} \\
    &= \frac{\beta^m}{2^{2m}} \sum^{m}_{k=0} \binom{2k}{k} \binom{2(m-k)}{m-k} 
    \left(\frac{\alpha}{\beta}\right)^k = \beta^m E_m(x),
\end{split}
\end{equation}
where $x = \alpha/\beta < 1$. We prove that the quantity
\begin{equation}
    E_m(x) = \frac{1}{2^{2m}} \sum^{m}_{k=0} \binom{2k}{k} \binom{2m-2k}{m-k} 
x^k
\end{equation}
monotonically decreases with $m$ for any fixed $x$ provided that $0 \leqslant x 
< 1$. 
Note that $E_m(1) = 1$. In fact, we have
\begin{equation}
    E_m(1) = [t^m] \frac{1}{\sqrt{1-t}} \frac{1}{\sqrt{1-t}} = [t^m] 
\frac{1}{1-t} = 1.
\end{equation}
Also note that $\binom{2n+2}{n+1} = (4 - 2/(n+1)) \binom{2n}{n}$ for $n \geqslant 0$. With these 
two facts established, we can now estimate $E_{m+1}(x)$ as $E_{m+1}(x) = E_m(x) - R_{m+1}(x)$, where
\begin{equation}
\begin{split}
    R_{m+1}(x) &= \frac{1}{2^{2m+2}} \left[2\sum^m_{k=0} 
\frac{\binom{2k}{k}}{k+1}  
    \binom{2(m-k)}{m-k} x^k \right.\\
    &- \left.\binom{2(m+1)}{m+1}x^{m+1}\right].
\end{split}
\end{equation}
Since $E_{m+1}(1) = E_m(1) = 1$, we get the equality
\begin{equation}
    2\sum^m_{k=0} \frac{1}{k+1} \binom{2k}{k} \binom{2(m-k)}{m-k} = 
    \binom{2(m+1)}{m+1}.
\end{equation}
According to our assumption, $0 \leqslant x < 1$, and thus we can estimate the first part of 
$R_{m+1}(x)$ as follows:
\begin{equation}
\begin{split}
    &2\sum^m_{k=0} \frac{1}{k+1} \binom{2k}{k} \binom{2(m-k)}{m-k} x^k \\
    &> 2\sum^m_{k=0} \frac{\binom{2k}{k}}{k+1} \binom{2(m-k)}{m-k} x^{m+1} = 
    \binom{2(m+1)}{m+1}x^{m+1}, \nonumber
\end{split}
\end{equation}
so that $R_{m+1}(x) > 0$ and thus $E_{m+1}(x) < E_m(x)$.

We now prove that the state $|\varphi_a\rangle$ with $a=1.5$ violates the inequalities for 
$S^{(m)}_3$ for all $m \geqslant 1$. We must prove that $\beta^m E_m(\alpha/\beta) < (2^m+2)/3$ for 
all $m$. For $m=1, 2, 3$ it is easy to prove this inequality by hand. For example, for $m=1$ it 
reads as $(71 - 2\sqrt{385})/24 < 4/3$, which can be equivalently transformed to the inequality $39 
< 2\sqrt{385}$ and finally to the valid inequality $1521 < 1540$. The other two cases can be 
established in a similar way, with integer arithmetic only. To prove the inequality for $m>3$ note 
that it is enough to prove a more strict inequality
\begin{equation}\label{eq:1}
    c \beta^m < \frac{2^m+2}{3}, \quad c = E_1\left(\frac{\alpha}{\beta}\right).
\end{equation}
If for some $m$ we show that $c(\beta-1)\beta^m < 2^m/3$ or, equivalently, that $3c(\beta-1) < 
(2/\beta)^m$, then the validity of the inequality \eqref{eq:1} for $m$ will imply its validity for 
$m+1$. It is possible to show that the inequality $3c(\beta-1) < (2/\beta)^m$ holds true for $m = 4$ 
and thus for all $m \geqslant 4$. We see that the violation grows as $(2/\beta)^m \approx 1.16^m$. 
This violation grows at a more moderate rate than that of the unphysical limit presented before, but 
it is attained for a real physical state.


\begin{thebibliography}{29}%
\makeatletter
\providecommand \@ifxundefined [1]{%
 \@ifx{#1\undefined}
}%
\providecommand \@ifnum [1]{%
 \ifnum #1\expandafter \@firstoftwo
 \else \expandafter \@secondoftwo
 \fi
}%
\providecommand \@ifx [1]{%
 \ifx #1\expandafter \@firstoftwo
 \else \expandafter \@secondoftwo
 \fi
}%
\providecommand \natexlab [1]{#1}%
\providecommand \enquote  [1]{``#1''}%
\providecommand \bibnamefont  [1]{#1}%
\providecommand \bibfnamefont [1]{#1}%
\providecommand \citenamefont [1]{#1}%
\providecommand \href@noop [0]{\@secondoftwo}%
\providecommand \href [0]{\begingroup \@sanitize@url \@href}%
\providecommand \@href[1]{\@@startlink{#1}\@@href}%
\providecommand \@@href[1]{\endgroup#1\@@endlink}%
\providecommand \@sanitize@url [0]{\catcode `\\12\catcode `\$12\catcode
  `\&12\catcode `\#12\catcode `\^12\catcode `\_12\catcode `\%12\relax}%
\providecommand \@@startlink[1]{}%
\providecommand \@@endlink[0]{}%
\providecommand \url  [0]{\begingroup\@sanitize@url \@url }%
\providecommand \@url [1]{\endgroup\@href {#1}{\urlprefix }}%
\providecommand \urlprefix  [0]{URL }%
\providecommand \Eprint [0]{\href }%
\providecommand \doibase [0]{http://dx.doi.org/}%
\providecommand \selectlanguage [0]{\@gobble}%
\providecommand \bibinfo  [0]{\@secondoftwo}%
\providecommand \bibfield  [0]{\@secondoftwo}%
\providecommand \translation [1]{[#1]}%
\providecommand \BibitemOpen [0]{}%
\providecommand \bibitemStop [0]{}%
\providecommand \bibitemNoStop [0]{.\EOS\space}%
\providecommand \EOS [0]{\spacefactor3000\relax}%
\providecommand \BibitemShut  [1]{\csname bibitem#1\endcsname}%
\let\auto@bib@innerbib\@empty
\bibitem [{\citenamefont {Doherty}\ \emph {et~al.}(2004)\citenamefont
  {Doherty}, \citenamefont {Parrilo},\ and\ \citenamefont
  {Spedalieri}}]{PhysRevA.69.022308}%
  \BibitemOpen
  \bibfield  {author} {\bibinfo {author} {\bibfnamefont {A.~C.}\ \bibnamefont
  {Doherty}}, \bibinfo {author} {\bibfnamefont {P.~A.}\ \bibnamefont
  {Parrilo}}, \ and\ \bibinfo {author} {\bibfnamefont {F.~M.}\ \bibnamefont
  {Spedalieri}},\ }\href@noop {} {\bibfield  {journal} {\bibinfo  {journal}
  {Phys. Rev. A}\ }\textbf {\bibinfo {volume} {69}},\ \bibinfo {pages} {022308}
  (\bibinfo {year} {2004})}\BibitemShut {NoStop}%
\bibitem [{\citenamefont {Peres}(1996)}]{PhysRevLett.77.1413}%
  \BibitemOpen
  \bibfield  {author} {\bibinfo {author} {\bibfnamefont {A.}~\bibnamefont
  {Peres}},\ }\href@noop {} {\bibfield  {journal} {\bibinfo  {journal} {Phys.
  Rev. Lett.}\ }\textbf {\bibinfo {volume} {77}},\ \bibinfo {pages} {1413}
  (\bibinfo {year} {1996})}\BibitemShut {NoStop}%
\bibitem [{\citenamefont {Horodecki}\ \emph {et~al.}(1996)\citenamefont
  {Horodecki}, \citenamefont {Horodecki},\ and\ \citenamefont
  {Horodecki}}]{Horodecki19961}%
  \BibitemOpen
  \bibfield  {author} {\bibinfo {author} {\bibfnamefont {M.}~\bibnamefont
  {Horodecki}}, \bibinfo {author} {\bibfnamefont {P.}~\bibnamefont
  {Horodecki}}, \ and\ \bibinfo {author} {\bibfnamefont {R.}~\bibnamefont
  {Horodecki}},\ }\href@noop {} {\bibfield  {journal} {\bibinfo  {journal}
  {Phys. Lett. A}\ }\textbf {\bibinfo {volume} {223}},\ \bibinfo {pages} {1}
  (\bibinfo {year} {1996})}\BibitemShut {NoStop}%
\bibitem [{\citenamefont {Werner}\ and\ \citenamefont
  {Wolf}(2001)}]{PhysRevLett.86.3658}%
  \BibitemOpen
  \bibfield  {author} {\bibinfo {author} {\bibfnamefont {R.~F.}\ \bibnamefont
  {Werner}}\ and\ \bibinfo {author} {\bibfnamefont {M.~M.}\ \bibnamefont
  {Wolf}},\ }\href@noop {} {\bibfield  {journal} {\bibinfo  {journal} {Phys.
  Rev. Lett.}\ }\textbf {\bibinfo {volume} {86}},\ \bibinfo {pages} {3658}
  (\bibinfo {year} {2001})}\BibitemShut {NoStop}%
\bibitem [{\citenamefont {Simon}(2000)}]{PhysRevLett.84.2726}%
  \BibitemOpen
  \bibfield  {author} {\bibinfo {author} {\bibfnamefont {R.}~\bibnamefont
  {Simon}},\ }\href@noop {} {\bibfield  {journal} {\bibinfo  {journal} {Phys.
  Rev. Lett.}\ }\textbf {\bibinfo {volume} {84}},\ \bibinfo {pages} {2726}
  (\bibinfo {year} {2000})}\BibitemShut {NoStop}%
\bibitem [{\citenamefont {Giedke}\ \emph {et~al.}(2001)\citenamefont {Giedke},
  \citenamefont {Kraus}, \citenamefont {Lewenstein},\ and\ \citenamefont
  {Cirac}}]{PhysRevA.64.052303}%
  \BibitemOpen
  \bibfield  {author} {\bibinfo {author} {\bibfnamefont {G.}~\bibnamefont
  {Giedke}}, \bibinfo {author} {\bibfnamefont {B.}~\bibnamefont {Kraus}},
  \bibinfo {author} {\bibfnamefont {M.}~\bibnamefont {Lewenstein}}, \ and\
  \bibinfo {author} {\bibfnamefont {J.~I.}\ \bibnamefont {Cirac}},\ }\href@noop
  {} {\bibfield  {journal} {\bibinfo  {journal} {Phys. Rev. A}\ }\textbf
  {\bibinfo {volume} {64}},\ \bibinfo {pages} {052303} (\bibinfo {year}
  {2001})}\BibitemShut {NoStop}%
\bibitem [{\citenamefont {van Loock}\ and\ \citenamefont
  {Furusawa}(2003)}]{PhysRevA.67.052315}%
  \BibitemOpen
  \bibfield  {author} {\bibinfo {author} {\bibfnamefont {P.}~\bibnamefont {van
  Loock}}\ and\ \bibinfo {author} {\bibfnamefont {A.}~\bibnamefont
  {Furusawa}},\ }\href@noop {} {\bibfield  {journal} {\bibinfo  {journal}
  {Phys. Rev. A}\ }\textbf {\bibinfo {volume} {67}},\ \bibinfo {pages} {052315}
  (\bibinfo {year} {2003})}\BibitemShut {NoStop}%
\bibitem [{\citenamefont {Aoki}\ \emph {et~al.}(2003)\citenamefont {Aoki},
  \citenamefont {Takei}, \citenamefont {Yonezawa}, \citenamefont {Wakui},
  \citenamefont {Hiraoka}, \citenamefont {Furusawa},\ and\ \citenamefont {van
  Loock}}]{PhysRevLett.91.080404}%
  \BibitemOpen
  \bibfield  {author} {\bibinfo {author} {\bibfnamefont {T.}~\bibnamefont
  {Aoki}}, \bibinfo {author} {\bibfnamefont {N.}~\bibnamefont {Takei}},
  \bibinfo {author} {\bibfnamefont {H.}~\bibnamefont {Yonezawa}}, \bibinfo
  {author} {\bibfnamefont {K.}~\bibnamefont {Wakui}}, \bibinfo {author}
  {\bibfnamefont {T.}~\bibnamefont {Hiraoka}}, \bibinfo {author} {\bibfnamefont
  {A.}~\bibnamefont {Furusawa}}, \ and\ \bibinfo {author} {\bibfnamefont
  {P.}~\bibnamefont {van Loock}},\ }\href@noop {} {\bibfield  {journal}
  {\bibinfo  {journal} {Phys. Rev. Lett.}\ }\textbf {\bibinfo {volume} {91}},\
  \bibinfo {pages} {080404} (\bibinfo {year} {2003})}\BibitemShut {NoStop}%
\bibitem [{\citenamefont {Duan}\ \emph {et~al.}(2000)\citenamefont {Duan},
  \citenamefont {Giedke}, \citenamefont {Cirac},\ and\ \citenamefont
  {Zoller}}]{PhysRevLett.84.2722}%
  \BibitemOpen
  \bibfield  {author} {\bibinfo {author} {\bibfnamefont {L.-M.}\ \bibnamefont
  {Duan}}, \bibinfo {author} {\bibfnamefont {G.}~\bibnamefont {Giedke}},
  \bibinfo {author} {\bibfnamefont {J.~I.}\ \bibnamefont {Cirac}}, \ and\
  \bibinfo {author} {\bibfnamefont {P.}~\bibnamefont {Zoller}},\ }\href@noop {}
  {\bibfield  {journal} {\bibinfo  {journal} {Phys. Rev. Lett.}\ }\textbf
  {\bibinfo {volume} {84}},\ \bibinfo {pages} {2722} (\bibinfo {year}
  {2000})}\BibitemShut {NoStop}%
\bibitem [{\citenamefont {Shalm}\ \emph {et~al.}(2013)\citenamefont {Shalm},
  \citenamefont {Hamel}, \citenamefont {Yan}, \citenamefont {Simon},
  \citenamefont {Resch},\ and\ \citenamefont {Jennewein}}]{NatPhys.9.19}%
  \BibitemOpen
  \bibfield  {author} {\bibinfo {author} {\bibfnamefont {L.~K.}\ \bibnamefont
  {Shalm}}, \bibinfo {author} {\bibfnamefont {D.~R.}\ \bibnamefont {Hamel}},
  \bibinfo {author} {\bibfnamefont {Z.}~\bibnamefont {Yan}}, \bibinfo {author}
  {\bibfnamefont {C.}~\bibnamefont {Simon}}, \bibinfo {author} {\bibfnamefont
  {K.~J.}\ \bibnamefont {Resch}}, \ and\ \bibinfo {author} {\bibfnamefont
  {T.}~\bibnamefont {Jennewein}},\ }\href@noop {} {\bibfield  {journal}
  {\bibinfo  {journal} {Nat. Phys.}\ }\textbf {\bibinfo {volume} {9}},\
  \bibinfo {pages} {19} (\bibinfo {year} {2013})}\BibitemShut {NoStop}%
\bibitem [{\citenamefont {Hyllus}\ and\ \citenamefont
  {Eisert}(2006)}]{NewJPhys.8.51}%
  \BibitemOpen
  \bibfield  {author} {\bibinfo {author} {\bibfnamefont {P.}~\bibnamefont
  {Hyllus}}\ and\ \bibinfo {author} {\bibfnamefont {J.}~\bibnamefont
  {Eisert}},\ }\href@noop {} {\bibfield  {journal} {\bibinfo  {journal} {New
  Journal of Physics}\ }\textbf {\bibinfo {volume} {8}},\ \bibinfo {pages} {51}
  (\bibinfo {year} {2006})}\BibitemShut {NoStop}%
\bibitem [{\citenamefont {Sperling}\ and\ \citenamefont
  {Vogel}(2013)}]{PhysRevLett.111.110503}%
  \BibitemOpen
  \bibfield  {author} {\bibinfo {author} {\bibfnamefont {J.}~\bibnamefont
  {Sperling}}\ and\ \bibinfo {author} {\bibfnamefont {W.}~\bibnamefont
  {Vogel}},\ }\href@noop {} {\bibfield  {journal} {\bibinfo  {journal} {Phys.
  Rev. Lett.}\ }\textbf {\bibinfo {volume} {111}},\ \bibinfo {pages} {110503}
  (\bibinfo {year} {2013})}\BibitemShut {NoStop}%
\bibitem [{\citenamefont {Ac\'{i}n}\ \emph {et~al.}(2001)\citenamefont
  {Ac\'{i}n}, \citenamefont {Bru\ss}, \citenamefont {Lewenstein},\ and\
  \citenamefont {Sanpera}}]{PhysRevLett.87.040401}%
  \BibitemOpen
  \bibfield  {author} {\bibinfo {author} {\bibfnamefont {A.}~\bibnamefont
  {Ac\'{i}n}}, \bibinfo {author} {\bibfnamefont {D.}~\bibnamefont {Bru\ss}},
  \bibinfo {author} {\bibfnamefont {M.}~\bibnamefont {Lewenstein}}, \ and\
  \bibinfo {author} {\bibfnamefont {A.}~\bibnamefont {Sanpera}},\ }\href@noop
  {} {\bibfield  {journal} {\bibinfo  {journal} {Phys. Rev. Lett.}\ }\textbf
  {\bibinfo {volume} {87}},\ \bibinfo {pages} {040401} (\bibinfo {year}
  {2001})}\BibitemShut {NoStop}%
\bibitem [{\citenamefont {Braunstein}(2005)}]{PhysRevA.71.055801}%
  \BibitemOpen
  \bibfield  {author} {\bibinfo {author} {\bibfnamefont {S.~L.}\ \bibnamefont
  {Braunstein}},\ }\href@noop {} {\bibfield  {journal} {\bibinfo  {journal}
  {Phys. Rev. A}\ }\textbf {\bibinfo {volume} {71}},\ \bibinfo {pages} {055801}
  (\bibinfo {year} {2005})}\BibitemShut {NoStop}%
\bibitem [{\citenamefont {Reck}\ \emph {et~al.}(1994)\citenamefont {Reck},
  \citenamefont {Zeilinger}, \citenamefont {Bernstein},\ and\ \citenamefont
  {Bertani}}]{PhysRevLett.73.58}%
  \BibitemOpen
  \bibfield  {author} {\bibinfo {author} {\bibfnamefont {M.}~\bibnamefont
  {Reck}}, \bibinfo {author} {\bibfnamefont {A.}~\bibnamefont {Zeilinger}},
  \bibinfo {author} {\bibfnamefont {H.~J.}\ \bibnamefont {Bernstein}}, \ and\
  \bibinfo {author} {\bibfnamefont {P.}~\bibnamefont {Bertani}},\ }\href@noop
  {} {\bibfield  {journal} {\bibinfo  {journal} {Phys. Rev. Lett.}\ }\textbf
  {\bibinfo {volume} {73}},\ \bibinfo {pages} {58} (\bibinfo {year}
  {1994})}\BibitemShut {NoStop}%
\bibitem [{\citenamefont {Adesso}\ and\ \citenamefont
  {Piano}(2014)}]{PhysRevLett.112.010401}%
  \BibitemOpen
  \bibfield  {author} {\bibinfo {author} {\bibfnamefont {G.}~\bibnamefont
  {Adesso}}\ and\ \bibinfo {author} {\bibfnamefont {S.}~\bibnamefont {Piano}},\
  }\href@noop {} {\bibfield  {journal} {\bibinfo  {journal} {Phys. Rev. Lett.}\
  }\textbf {\bibinfo {volume} {112}},\ \bibinfo {pages} {010401} (\bibinfo
  {year} {2014})}\BibitemShut {NoStop}%
\bibitem [{\citenamefont {Olivares}\ and\ \citenamefont
  {Paris}(2008)}]{EPJSP.160.319}%
  \BibitemOpen
  \bibfield  {author} {\bibinfo {author} {\bibfnamefont {S.}~\bibnamefont
  {Olivares}}\ and\ \bibinfo {author} {\bibfnamefont {M.~G.}\ \bibnamefont
  {Paris}},\ }\href@noop {} {\bibfield  {journal} {\bibinfo  {journal} {Eur.
  Phys. J. Special Topics}\ }\textbf {\bibinfo {volume} {160}},\ \bibinfo
  {pages} {319} (\bibinfo {year} {2008})}\BibitemShut {NoStop}%
\bibitem [{\citenamefont {Shchukin}\ and\ \citenamefont
  {Vogel}(2006{\natexlab{a}})}]{PhysRevA.74.030302}%
  \BibitemOpen
  \bibfield  {author} {\bibinfo {author} {\bibfnamefont {E.}~\bibnamefont
  {Shchukin}}\ and\ \bibinfo {author} {\bibfnamefont {W.}~\bibnamefont
  {Vogel}},\ }\href@noop {} {\bibfield  {journal} {\bibinfo  {journal} {Phys.
  Rev. A}\ }\textbf {\bibinfo {volume} {74}},\ \bibinfo {pages} {030302}
  (\bibinfo {year} {2006}{\natexlab{a}})}\BibitemShut {NoStop}%
\bibitem [{\citenamefont {Shchukin}\ and\ \citenamefont
  {Vogel}(2006{\natexlab{b}})}]{PhysRevLett.96.200403}%
  \BibitemOpen
  \bibfield  {author} {\bibinfo {author} {\bibfnamefont {E.}~\bibnamefont
  {Shchukin}}\ and\ \bibinfo {author} {\bibfnamefont {W.}~\bibnamefont
  {Vogel}},\ }\href@noop {} {\bibfield  {journal} {\bibinfo  {journal} {Phys.
  Rev. Lett.}\ }\textbf {\bibinfo {volume} {96}},\ \bibinfo {pages} {200403}
  (\bibinfo {year} {2006}{\natexlab{b}})}\BibitemShut {NoStop}%
\bibitem [{\citenamefont {Allevi}\ \emph {et~al.}(2012)\citenamefont {Allevi},
  \citenamefont {Olivares},\ and\ \citenamefont
  {Bondani}}]{PhysRevA.85.063835}%
  \BibitemOpen
  \bibfield  {author} {\bibinfo {author} {\bibfnamefont {A.}~\bibnamefont
  {Allevi}}, \bibinfo {author} {\bibfnamefont {S.}~\bibnamefont {Olivares}}, \
  and\ \bibinfo {author} {\bibfnamefont {M.}~\bibnamefont {Bondani}},\
  }\href@noop {} {\bibfield  {journal} {\bibinfo  {journal} {Phys. Rev. A}\
  }\textbf {\bibinfo {volume} {85}},\ \bibinfo {pages} {063835} (\bibinfo
  {year} {2012})}\BibitemShut {NoStop}%
\bibitem [{\citenamefont {Nha}\ \emph {et~al.}(2012)\citenamefont {Nha},
  \citenamefont {Lee}, \citenamefont {Ji},\ and\ \citenamefont
  {Kim}}]{PhysRevLett.108.030503}%
  \BibitemOpen
  \bibfield  {author} {\bibinfo {author} {\bibfnamefont {H.}~\bibnamefont
  {Nha}}, \bibinfo {author} {\bibfnamefont {S.-Y.}\ \bibnamefont {Lee}},
  \bibinfo {author} {\bibfnamefont {S.-W.}\ \bibnamefont {Ji}}, \ and\ \bibinfo
  {author} {\bibfnamefont {M.~S.}\ \bibnamefont {Kim}},\ }\href@noop {}
  {\bibfield  {journal} {\bibinfo  {journal} {Phys. Rev. Lett.}\ }\textbf
  {\bibinfo {volume} {108}},\ \bibinfo {pages} {030503} (\bibinfo {year}
  {2012})}\BibitemShut {NoStop}%
\bibitem [{Note1()}]{Note1}%
  \BibitemOpen
  \bibinfo {note} {There is an alternative way of proving this theorem that
  does not use PT explicitly. But nevertheless it turns out that the theorem is
  PT-related --- if a state is PPT then the inequality \protect \textup {\hbox
  {\mathsurround \z@ \protect \normalfont (\ignorespaces \ref {eq:ABn}\unskip
  \@@italiccorr )}} cannot be violated.}\BibitemShut {Stop}%
\bibitem [{\citenamefont {Adesso}\ \emph {et~al.}(2006)\citenamefont {Adesso},
  \citenamefont {Serafini},\ and\ \citenamefont
  {Illuminati}}]{PhysRevA.73.032345}%
  \BibitemOpen
  \bibfield  {author} {\bibinfo {author} {\bibfnamefont {G.}~\bibnamefont
  {Adesso}}, \bibinfo {author} {\bibfnamefont {A.}~\bibnamefont {Serafini}}, \
  and\ \bibinfo {author} {\bibfnamefont {F.}~\bibnamefont {Illuminati}},\
  }\href@noop {} {\bibfield  {journal} {\bibinfo  {journal} {Phys. Rev. A}\
  }\textbf {\bibinfo {volume} {73}},\ \bibinfo {pages} {032345} (\bibinfo
  {year} {2006})}\BibitemShut {NoStop}%
\bibitem [{\citenamefont {Bowen}\ \emph {et~al.}(2003)\citenamefont {Bowen},
  \citenamefont {Lam},\ and\ \citenamefont {Ralph}}]{JModOpt.50.801}%
  \BibitemOpen
  \bibfield  {author} {\bibinfo {author} {\bibfnamefont {W.~P.}\ \bibnamefont
  {Bowen}}, \bibinfo {author} {\bibfnamefont {P.~K.}\ \bibnamefont {Lam}}, \
  and\ \bibinfo {author} {\bibfnamefont {T.~C.}\ \bibnamefont {Ralph}},\
  }\href@noop {} {\bibfield  {journal} {\bibinfo  {journal} {J. Mod. Opt}\
  }\textbf {\bibinfo {volume} {50}},\ \bibinfo {pages} {801} (\bibinfo {year}
  {2003})}\BibitemShut {NoStop}%
\bibitem [{\citenamefont {van Loock}\ and\ \citenamefont
  {Braunstein}(2000)}]{PhysRevLett.84.3482}%
  \BibitemOpen
  \bibfield  {author} {\bibinfo {author} {\bibfnamefont {P.}~\bibnamefont {van
  Loock}}\ and\ \bibinfo {author} {\bibfnamefont {S.~L.}\ \bibnamefont
  {Braunstein}},\ }\href@noop {} {\bibfield  {journal} {\bibinfo  {journal}
  {Phys. Rev. Lett.}\ }\textbf {\bibinfo {volume} {84}},\ \bibinfo {pages}
  {3482} (\bibinfo {year} {2000})}\BibitemShut {NoStop}%
\bibitem [{\citenamefont {Kim}\ \emph {et~al.}(2002)\citenamefont {Kim},
  \citenamefont {Lee},\ and\ \citenamefont {Munro}}]{PhysRevA.66.030301}%
  \BibitemOpen
  \bibfield  {author} {\bibinfo {author} {\bibfnamefont {M.~S.}\ \bibnamefont
  {Kim}}, \bibinfo {author} {\bibfnamefont {J.}~\bibnamefont {Lee}}, \ and\
  \bibinfo {author} {\bibfnamefont {W.~J.}\ \bibnamefont {Munro}},\ }\href@noop
  {} {\bibfield  {journal} {\bibinfo  {journal} {Phys. Rev. A}\ }\textbf
  {\bibinfo {volume} {66}},\ \bibinfo {pages} {030301} (\bibinfo {year}
  {2002})}\BibitemShut {NoStop}%
\bibitem [{\citenamefont {Mandel}\ and\ \citenamefont
  {Wolf}(1995)}]{mandel-wolf}%
  \BibitemOpen
  \bibfield  {author} {\bibinfo {author} {\bibfnamefont {L.}~\bibnamefont
  {Mandel}}\ and\ \bibinfo {author} {\bibfnamefont {E.}~\bibnamefont {Wolf}},\
  }\href@noop {} {\emph {\bibinfo {title} {Optical coherence and quantum
  optics}}}\ (\bibinfo  {publisher} {Cambridge},\ \bibinfo {year}
  {1995})\BibitemShut {NoStop}%
\bibitem [{\citenamefont {Erd\'{e}lyi}(1953)}]{BE2-2}%
  \BibitemOpen
  \bibinfo {editor} {\bibfnamefont {A.}~\bibnamefont {Erd\'{e}lyi}},\ ed.,\
  \href@noop {} {\emph {\bibinfo {title} {Higher transcendental functions}}},\
  Vol.~\bibinfo {volume} {2}\ (\bibinfo {year} {1953})\BibitemShut {NoStop}%
\bibitem [{\citenamefont {Horn}\ and\ \citenamefont
  {Johnson}(2013)}]{horn-johnson}%
  \BibitemOpen
  \bibfield  {author} {\bibinfo {author} {\bibfnamefont {R.~A.}\ \bibnamefont
  {Horn}}\ and\ \bibinfo {author} {\bibfnamefont {C.~R.}\ \bibnamefont
  {Johnson}},\ }\href@noop {} {\emph {\bibinfo {title} {Matrix analysis}}},\
  \bibinfo {edition} {2nd}\ ed.\ (\bibinfo  {publisher} {Cambridge},\ \bibinfo
  {year} {2013})\BibitemShut {NoStop}%
\end{thebibliography}
\end{document}